\documentclass[runningheads]{llncs}
\usepackage{graphicx}

\usepackage{amsmath,amssymb}
\usepackage{hyperref}

\usepackage[ruled]{algorithm2e}

\usepackage{thmtools}
\usepackage{thm-restate}
\usepackage[capitalise]{cleveref}

\usepackage{crossreftools}
\hypersetup{%
    bookmarksnumbered, bookmarksopen=true, bookmarksopenlevel=1,%
}

\usepackage{booktabs}
\usepackage{adjustbox}

\usepackage{stmaryrd}

\usepackage{listings}
\usepackage{float}
\usepackage[bold]{complexity}

\newcommand{\porous}{\textsc{porous}}

\newcommand{\naturals}{\mathbb{N}}
\newcommand{\N}{\mathbb{N}}

\newcommand{\ints}{\mathbb{Z}}

\newcommand{\reals}{\mathbb{R}}
\newcommand{\rp}{\mathbb{R}_+}

\newcommand{\set}[1]{\left\{#1\right\}}

\newcommand{\vecit}[2]{#1^{(#2)}}

\newcommand{\target}{y}

\newcommand{\I}{\mathcal{I}}

\def\eps{\varepsilon}
\newcommand{\enc}[1]{[ #1 ]}

\usepackage{multirow}

\usepackage{xcolor}

\newcounter{SideNoteCounter} \stepcounter{SideNoteCounter}

\newcommand{\lcm}{\operatorname{lcm}}

\declaretheorem[name={Meta Problem},Refname={Meta Problem,Meta Problems},refname={Meta Problem, Meta Problems}]{metaproblem}

\newcommand{\aff}[1]{\overline{#1}^a}
\newcommand{\zar}[1]{\overline{#1}^z}
\newcommand{\fuss}[1]{\overline{#1}^\reals}
\newcommand{\orbit}{\mathcal{O}}

\usepackage[firstpage]{draftwatermark}

\SetWatermarkText{\hspace*{5in}\raisebox{5.in}{\includegraphics[scale=0.8]{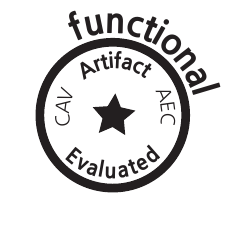}}}

\SetWatermarkAngle{0}
\begin{document}
\title{Porous Invariants%
}
\author{Engel Lefaucheux\inst{1} 
\and Joël Ouaknine\inst{1} 
\and David Purser\inst{1} 
\and James Worrell\inst{2} 
} 
\authorrunning{E. Lefaucheux et al.}

\institute{Max Planck Institute for Software Systems, Saarland
  Informatics Campus, Saarbr\"ucken, Germany \and
Department of Computer Science, Oxford University, UK\\
}
\maketitle              %
\begin{abstract}
We introduce the notion of \emph{porous invariants} for multipath (or
branching/nondeterministic) affine loops over the integers; these
invariants are not necessarily convex, and can in fact contain
infinitely many `holes'. Nevertheless, we show that in many cases
such invariants can be automatically synthesised, and moreover can be
used to settle (non-)reachability questions for various interesting classes
of affine loops and target sets.

\end{abstract}
\section{Introduction}
We consider the reachability problem for multipath (or branching) affine loops over the integers, or equivalently for nondeterministic integer linear dynamical systems. A (deterministic) integer linear dynamical system consists of an update matrix $M \in \ints^{d\times d}$ together with an initial point $\vecit{x}{0} \in \ints^d$. We associate to such a system its infinite orbit $(\vecit{x}{i})$ consisting of the sequence of reachable points defined by the rule $\vecit{x}{i+1} = M\vecit{x}{i}$.
The reachability question then asks, given a target set $Y$, whether the orbit ever meets $Y$, i.e., whether 
there exists some time $i$ such that $\vecit{x}{i} \in Y$. The nondeterministic reachability question allows the linear update map to be chosen at each step from a fixed finite collection of matrices.

When the orbit does eventually hit the target, one can easily substantiate this by exhibiting the relevant finite prefix. However, establishing non-reachability is intrinsically more difficult, since the orbit consists of an infinite sequence of points. One requires some sort of finitary certificate, which must be a relatively simple object that can be inspected and which provides a proof that the set $Y$ is indeed unreachable. Typically, such a certificate will consist of an over-approximation $I$ of the set $R$ of reachable points, in such a manner that one can check both that $Y \cap I=\emptyset$ and $R\subseteq I$; such a set $I$ is called an invariant.

Formally we study the following problem for \textit{inductive invariants}:

\begin{metaproblem}\label{metaproblem} Consider a system with update functions $f_1,\dots, f_n$. A set $I$ is an inductive invariant if $f_i(I)\subseteq I$ for all $i$. Given a reachability query $(x,Y)$ we search for a separating inductive invariant $I$ such that $x\in I$ and $Y\cap I = \emptyset$.\end{metaproblem}

\cref{metaproblem} is parametrised by the type of invariants and targets that are considered; that is, what are the classes of allowable invariant sets $I$ and target sets $Y$, or equivalently how are such sets allowed to be expressed. 

Fixing a particular invariant and target domain, a reachability query has three possible scenarios: (1) the instance is reachable, (2) the instance is unreachable and a separating invariant from the domain exists, or (3) the instance is unreachable but no separating invariant exists. Ideally, one would wish to provide a sufficiently expressive invariant domain so that the latter case does not occur, whilst keeping the resulting invariants as simple as possible and computable. For some classes of systems, it is known that distinguishing reachability (1) from unreachability (2,3) is undecidable; it can also happen that
determining whether a separating invariant exists (i.e., distinguishing (2) from (3)) is undecidable.

We note that the existence of \emph{strongest} inductive invariants\footnote{Given two invariants $I$ and $I'$, we say that $I$ is \emph{stronger} than $I'$ iff $I \subseteq I'$; thus \emph{strongest} invariants correspond to \emph{smallest} invariant sets.} is a desirable property for an invariant domain---when strongest invariants exist (and can be computed), separating (2) from (1,3) is easy: compute the strongest invariant, and check whether it excludes the target state or not; if so, then you are done, and if not, no other invariant (from that class) can possibly do the trick either. However, unless (3) is excluded, computing the strongest invariant does not necessarily imply that reachability is decidable.
Unfortunately, strongest invariants are not always guaranteed to exist for a particular invariant domain, although some separating inductive invariant may still exist for every target (or indeed may not).

In prior work from the literature, typical classes of invariants are usually convex, or finite unions of convex sets. In this paper we consider certain classes of invariants that can have infinitely many `holes'
(albeit in a structured and regular way); we call such sets \emph{porous invariants}. These invariants can be represented via Presburger arithmetic\footnote{Presburger arithmetic is a decidable theory over the natural numbers, comprising Boolean operations, first-order quantification, and addition (but not multiplication).}. We shall work instead with the equivalent formulation of semi-linear sets, generalising ultimately periodic sets to higher dimensions, as finite unions of linear sets of the form $\set{b + p_1\mathbb{N} + \dots + p_m\mathbb{N}}$ (by which we mean $\set{b + a_1p_1 + \dots + a_mp_m\mid a_1,\dots,a_m\in\mathbb{N}}$, see \cref{defn:semi-linear}).

Let us first consider a motivating example:
\begin{example}[Hofstadter's MU Puzzle~\cite{douglas1979godel}]\label{mu:puzzle}
Consider the following term-rewriting puzzle over alphabet $\{M,U,I\}$. Start with the word $MI$, and by applying the following grammar rules (where $y$ and $z$ stand for arbitrary words over our alphabet), we ask whether the word $MU$ can ever be reached. 
\[
yI \to yIU \quad | \quad My \to Myy \quad | \quad yIIIz \to yUz \quad | \quad yUUz \to yz
\]

The answer is \emph{no}. One way to establish this is to keep track of the number of occurrences of the letter `$I$' in the words that can be produced, and observe that this number (call it $x$) will always be congruent to either $1$ or $2$ modulo $3$. In other words, it is not possible to reach the set
$\{x \mid x \equiv 0 \mod 3\}$. Indeed, Rules~$2$ and $3$ are the only rules that affect the number of $I$'s, and can be described by the system dynamics $x \mapsto 2x$ and $x \mapsto x -3$. Hence
the MU Puzzle can be viewed as a one-dimensional system with two affine updates,\footnote{One-dimensional affine updates are functions of the form $f(x) = ax + b$.} or a two-dimensional system with two linear updates.\footnote{$\begin{pmatrix} a & b\\ 0 & 1 \end{pmatrix}\begin{pmatrix}x \\ 1\end{pmatrix} = \begin{pmatrix}ax + b \\ 1\end{pmatrix}$ models affine functions using a matrix representation, holding one of the entries fixed to 1.} The set  $\set{1 + 3\ints}\cup \set{2 + 3\ints}$ is an inductive invariant, and we wish to synthesise this. (The stability of this set under our two affine functions is easily checked: both components are invariant under $x\mapsto x-3$, and $\set{1 + 3\ints}\mapsto \set{2 + 6\ints} \subseteq \set{2+3\ints}$ under $x\mapsto 2x$, and similarly $\set{2 + 3\ints}\mapsto \set{4 + 6\ints} \subseteq \set{1+3\ints}$.)

The problem can be rephrased as a safety property of the following multipath loop, verifying that the `bad' state $x= 0$ is never reached, or equivalently that the above loop can never halt, regardless of the nondeterministic choices made.\\
{\tt 
x $ = 1$\\
while x $\ne 0$\\
\indent x $ = 2$ x  $\mid\mid$  x $=$ x$-3$} $\qquad$(where $\mid\mid$ represents nondeterministic branching)\\

The MU Puzzle was presented as a challenge for algorithmic verification in \cite{ClarkeFHKOST03}; the tools considered in that paper (and elsewhere, to the best of our knowledge) rely upon the manual provision of an abstract invariant template. Our approach is to find the invariant fully automatically (although one must still abstract from the MU Puzzle the correct formulation as the program $x \mapsto 2x \mid\mid x \mapsto x-3$).
\end{example}

\subsubsection*{Main Contributions.}
Our focus is on the automatic generation of porous invariants for multipath affine loops over the integers, or equivalently nondeterministic integer linear dynamical systems.

\begin{itemize}
 \item We first consider targets consisting of a single vector (or `point targets'), and present the classes of invariants and systems for which invariants can and cannot be automatically computed for the reachability question. A summary of the results for linear and semi-linear invariants for these targets is given in \cref{fig:tableresults}. For completeness we also consider $\reals,\rp$-(semi)-linear sets, where we complete the picture from prior work by showing that strongest $\reals$-semi-linear invariants are computable.
 \begin{itemize}
 \item We establish the existence of \emph{strongest} $\ints$-linear invariants, and show that they can be found algorithmically %
   (\cref{thm:smallestZlinear}). These invariants may or may not separate the target under consideration. 
        \item If a $\ints$-linear invariant is not separating, we may instead look for an $\naturals{}$-semi-linear invariant (which generalises both $\ints$-semi-linear and $\naturals{}$-linear invariants), and we show that such an invariant can always be found for any
          unreachable point target when dealing with \emph{deterministic} integer linear dynamical systems (\cref{thm:nsemi-linear}).
 \item However, for nondeterministic integer linear dynamical systems, computing an $\naturals{}$-semi-linear invariants is an undecidable problem in arbitrary dimension (\cref{thm:undec}). Nevertheless we show how such invariants can be constructed in a low-dimensional setting,  in particular for affine updates in one dimension (\cref{thm:1daffine}).
   As an immediate consequence, this establishes that the multipath loop associated with the MU Puzzle belongs to a class of programs for which we can automatically synthesise $\naturals{}$-semi-linear invariants.

\end{itemize}
\item For \textit{full-dimensional}\footnote{The affine span covers the entire space.} $\ints$-linear targets we show that reachability is decidable, and, in the case of unreachability that a $\ints$-semi-linear invariant can always be exhibited as a certificate (\cref{thm:ztargets}). If the target is \emph{not} full-dimensional then the reachability problem is Skolem-hard and undecidable for deterministic and nondeterministic systems respectively.
\item In~\cref{sec:tool} we present our tool \porous\ which handles one-dimensional affine systems for both point and $\ints$-linear targets, solving both the reachability problem and producing invariants. Inter alia, this allows one to handle the multipath loop derived from the MU Puzzle in fully automated manner.
\end{itemize}

\begin{table}[t]
\centering
\begin{adjustbox}{width=\textwidth,center}
\begin{tabular}{|l|l|l|l|}
\hline
Dom &D/N & \textbf{Linear} & \textbf{Semi-linear} (SL) \\ \hline
  $\ints$ & det& Strongest computable %
                 (Thm.~\ref{thm:smallestZlinear}) & No strongest (Sec.~\ref{sec:no-strongest}); subsumed by $\naturals$-SL\\ \hline
  $\ints$ & non& Strongest computable %
                 (Thm.~\ref{thm:smallestZlinear}) & No strongest (Sec.~\ref{sec:no-strongest})  \\ \hline
$\naturals$ & det & No strongest (Sec.~\ref{sec:no-strongest}); subsumed by $\naturals$-SL & No strongest (Sec.~\ref{sec:no-strongest}), but sufficient computable (Thm.~\ref{thm:nsemi-linear})  \\ \hline 
$\naturals$ &non & No strongest (Sec.~\ref{sec:no-strongest}) & 1d-affine decidable (Thm.~\ref{thm:1daffine}); undecidable in general (Thm.~\ref{thm:undec})  \\ \hline 
$\mathbb{R}$ &det &  Strongest: affine relations by Karr~\cite{Karr76} &Strongest: affine closure on Zariski closure (Thm.~\ref{thm:fusses}) \\ \hline
$\mathbb{R}$ &non &  Strongest: affine relations by Karr~\cite{Karr76} &Strongest: affine closure on Zariski closure (Thm.~\ref{thm:fusses}) \\ \hline
$\mathbb{R}_+$ & det & No strongest (Sec.~\ref{sec:no-strongest}); subsumed by $\rp$-SL & No strongest, but sufficient computable~\cite{fijalkow2019monniaux} \\\hline
$\mathbb{R}_+$ &non& No strongest (Sec.~\ref{sec:no-strongest}) & Undecidable~\cite{fijalkow2019monniaux}\\\hline

\end{tabular}
\end{adjustbox}
\caption{Results for integer linear dynamical systems for a point target. Det/Non refers to deterministic or nondeterministic LDS\@. ``Subsumed by \ldots'' means that sufficient invariants can be generated, but of a more general type.}
\label{fig:tableresults}

\end{table}

\subsection{Related Work}
The reachability problem (in arbitrary dimension) for loops with a single affine update, or equivalently for deterministic linear dynamical systems, is decidable in polynomial time for point targets (that is $Y = \set{\target}$), as shown by Kannan and Lipton~\cite{KannanL86}.
However for nondeterministic systems (where the update matrix is chosen nondeterministically
from a finite set at each time step), reachability is undecidable, by 
reduction from the matrix semigroup membership problem~\cite{markov1947certain}.

In particular this entails that for unreachable nondeterministic instances we cannot hope \emph{always} to be able to compute a separating invariant. In some cases we may compute the strongest invariant (which may suffice if this invariant happens to be separating for the given reachability query), or we may compute an invariant in sub-cases for which reachability is decidable (for example in low dimensions). For some classes of invariants, it is also undecidable whether an invariant exists (e.g., polyhedral invariants~\cite{fijalkow2019monniaux}).

Various types of invariants have been studied for linear dynamical systems, including polyhedra~\cite{Monniaux19,fijalkow2019monniaux}, algebraic~\cite{HrushovskiOP018}, and o-minimal~\cite{almagor2019minimal} invariants. For certain classes of invariants (e.g., algebraic~\cite{HrushovskiOP018}), it is decidable whether a separating invariant exists, notwithstanding the reachability problem being undecidable.
Other works (e.g.,~\cite{CousotH78}) use heuristic approaches to generate invariants, without aiming for any sort of completeness.

Kincaid, Breck, Cyphert and Reps~\cite{KincaidBCR19} study loops with linear updates, studying the closed forms for the variables to prove safety and termination properties. Such closed forms, when expressible in certain arithmetic theories, can be interpreted as another type of invariant and can be used to over-approximate the reachable sets. The work is restricted to a single update function (deterministic loops) and places additional constraints on the updates to bring the closed forms into appropriate theories.

Bozga, Iosif and Konecn{\'{y}}'s  FLATA tool~\cite{BozgaIK10} considers affine functions in arbitrary dimension. However, it is restricted to affine functions with finite monoids; in our one-dimensional case this would correspond to limiting oneself to counter-like functions of the form $f(x) = x+b$.

Finkel, G{\"{o}}ller and Haase~\cite{FinkelGH13}, extending Fremont~\cite{Fremont2013}, show that reachability in a single dimension is $\PSPACE$-complete for polynomial update functions (and allowing states can be used to control the sequences of updates which can be applied).  The affine functions (and single-state restriction) we consider are a special case, but we focus on producing invariants to disprove reachability.

Other tools, e.g., \textsc{AProVE}~\cite{GieslABEFFHOPSS17} and B\"uchi Automizer~\cite{HeizmannHP14} may (dis-)prove termination/reachability on \textit{all} branches, but may not be able to prove termination/reachability on \textit{some} branch.

Inductive invariants specified in Presburger arithmetic have been used to disprove reachability in vector addition systems~\cite{Leroux10}. A generalisation, `almost semi-linear sets'~\cite{Leroux11} are also non-convex and can capture exactly the reachable points of vector addition systems. Our nondeterministic linear dynamical systems can be seen as vector addition systems over $\ints$ extended with affine updates (rather than only additive updates).

\section{Preliminaries}
We denote by $\ints{}$ the integers and $\naturals{}$ the non-negative integers.  We say that $x,y \in\ints$ are congruent modulo $d\in \naturals$, denoted $x \equiv y \mod d$, if $d$ divides $x-y$. Given an integer $x$ and natural $d$ we write $(x \mod d)$ for the number in $\set{0,\dots,d-1}$ such that $(x\mod d) \equiv x \mod d $.

\begin{definition}[Integer Linear Dynamical Systems]
A $d$-dimensional integer linear dynamical system (LDS) $(\vecit{x}{0},\{M_1,\dots,M_k\})$ is defined by an initial point $\vecit{x}{0}\in \ints{}^{d}$ and a set of integer matrices $M_1, \dots, M_k\subseteq \ints^{d\times d}$. An LDS is \emph{deterministic} if it comprises a single matrix ($k=1$) and is otherwise \emph{nondeterministic}.

A point $\target$ is \emph{reachable} if there exists $m\in\naturals$ and $B_1,\dots,B_m$ such that $B_1 \cdots  B_m \vecit{x}{0} = y$ and $B_i\in \set{M_1,\dots, M_k}$ for all $1\le i \le m$.

The \emph{reachability set} $\orbit \subseteq \ints{}^{d}$ of an LDS is the set of reachable points.
\end{definition}

\begin{definition}[$\mathbb{K}$-(semi)-linear sets]\label{defn:semi-linear}
A \emph{linear set} $L$ is defined by a base vector $b\in \ints^d$ and period vectors $p_1,\dots,p_d\in \ints^d$ such that  \[ L = 
\set{b+  a_1p_1+\dots + a_dp_d \mid a_1,\dots,a_d\in \mathbb{K}}.
\] 
For convenience we often write $\set{b + p_1\mathbb{K}+\dots + p_d\mathbb{K}}$ for $L$. A set is \emph{semi-linear} if it is the finite union of linear sets.
\end{definition}

$\naturals$-semi-linear sets are precisely those definable in Presburger arithmetic\linebreak ($\operatorname{FO}(\mathbb{Z}, + , \le)$)~\cite{gins}. However, we can also consider $\ints$-semi-linear sets (corresponding to $\operatorname{FO}(\mathbb{Z}, +)$ without order), and the real counterparts ($\reals$ and $\rp$). Note that even if $\mathbb{K}= \naturals$ we still allow $p_i \in \ints^d$.

\begin{definition}
Given an integer linear dynamical system $(\vecit{x}{0},\{M_1,\dots,M_k\})$, a set $I$ is an \emph{inductive invariant} if
\begin{itemize}
	\item $\vecit{x}{0}\in I$, and
	\item $\{M_i x \mid x \in I\}\subseteq I$ for all $i \in\set{1,\dots,k}$.
\end{itemize}
\end{definition}

Note in particular that every inductive invariant contains the reachability set ($\orbit \subseteq I$). 
We are interested in the following problem:
\begin{definition}[Invariant Synthesis Problem]
Given an invariant domain $\mathcal{D}$, an integer linear dynamical system $(\vecit{x}{0},\{M_1,\dots,M_k\})$, and a target $Y$, does there exist an inductive invariant $I$ in $\mathcal{D}$ disjoint from $Y$?
\end{definition}
In our setting, we are interested in classes $\mathcal{D}$ of invariants that are linear, or semi-linear. When a separating inductive invariant $I$ exists, we also wish to compute it.  Since (semi)-linear invariants are enumerable, the decision problem is, in theory, sufficient---although all of our proofs are constructive.

\section{\texorpdfstring{$\reals$}{R} Invariants: \texorpdfstring{$\reals$}{R}-linear and \texorpdfstring{$\reals$}{R}-semi-linear}
Before delving into porous invariants, let us consider invariants over the real numbers, i.e., described as $\reals$-(semi)-linear sets.

Strongest $\reals$-linear invariants are given precisely by the affine hull of the reachability set, and can be computed using Karr's algorithm~\cite{Karr76}. Moreover, we will show that strongest $\reals{}$-semi-linear invariants also exist and can be computed by combining  techniques for algebraic invariants~\cite{HrushovskiOP018} and $\reals{}$-linear invariants.

\paragraph*{$\reals$-linear.}
Recall that a set $L$ is $\mathbb{R}$-linear if $L = \set{v_{0} +v_{1}\reals+\dots+v_{t}\reals}$ for some $v_{0},\dots,v_{t} \in\ints^d$ that can be assumed to be linearly-independent\footnote{$v_0,\dots,v_m$ are linearly independent if there does not exist $a_0,\dots,a_m\in\mathbb{R}$, not all $0$, such that $a_0v_0+\dots+a_mv_m = 0$.} without loss of generality (and thus $t\le d$). Given two distinct points of $L$, every point on the infinite line connecting them must also be in $L$. Generalising this idea to higher dimensions, given a set $S\subseteq \mathbb{R}^d$, let the affine hull be \[\aff{S}= \set{\sum_{i=1}^k \lambda_i x_i \mid k\in\naturals, x_i \in S,\lambda_i \in \mathbb{R}, \sum_{i=1}^k \lambda_i = 1 }.\] 

Fix an LDS $(\vecit{x}{0},\{M_1,\dots,M_k\})$ and consider its reachability set $\orbit=$\linebreak
$\set{ M_{i_m}\cdots M_{i_1}\vecit{x}{0}  \mid  m \in \mathbb{N}, i_1,\ldots,i_m \in \set{1,\ldots,k} }$. 
Then $\aff{\orbit{}}$ is precisely the\linebreak strongest $\mathbb{R}$-linear invariant. Karr's algorithm~\cite{Karr76,Tzeng92} can be used to compute this strongest invariant in polynomial time. The next lemma follows from Theorem~3.1 of~\cite{Tzeng92}.
\begin{restatable}{lemma}{thmtzeng}\label{thm:tzeng}Given an LDS $(\vecit{x}{0}, \set{M_1,\dots, M_k})$ of dimension $d$, we can compute in time polynomial in $d$, $k$,
and $\log \mu$ (where $\mu>0$ is an upper bound on the absolute values of the integers appearing in $\vecit{x}{0}$ and $M_1,\ldots,M_k$), a $\mathbb{Q}$-affinely independent set of integer vectors $R_0 \subseteq \orbit{}$ such that:
\begin{enumerate}
\item  $\vecit{x}{0}\in R_0$,
\item the affine span of $R_0$ and the affine span of $\orbit{}$ are the same ($\aff{R_0} = \aff{\orbit{}}$), 
\item the entries of the vectors in $R_0$ have absolute value at most
  $\mu_0:=(d\mu)^d$.
\end{enumerate}
\end{restatable}
Let $R_0 = \set{\vecit{x}{0}, r_1,\dots,r_{d'}}$ be obtained as per 
\cref{thm:tzeng}, with $d' \le d$. The $\reals$-linear invariant of the LDS is the affine span $\aff{R_0}$, which can be written as the $\reals$-linear  set $L_0 = \set{\vecit{x}{0} + (r_1 - \vecit{x}{0})\reals + \dots + (r_{d'} - \vecit{x}{0})\reals}$. 

\paragraph*{$\reals$-semi-linear.} Let us now generalise this approach to $\reals$-semi-linear sets. The  collection of $\mathbb{R}$-semi-linear sets, $\set{\bigcup_{i = 1}^m L_i \mid m\in\naturals, L_1,\dots, L_m \text{ are } \mathbb{R}\text{-linear sets} }$, is closed under finite unions and arbitrary intersections\footnote{When intersecting a linear set with a semi-linear set, either the latter does not change, or one obtains a finite union of elements of smaller dimension. Thus, in an infinite intersection, only a finite number of intersections affects the original set.}. Thus for any given set $X$, the smallest $\mathbb{R}$-semi-linear set containing $X$ is simply the intersection of all $\mathbb{R}$-semi-linear sets containing $X$. 
Let us denote by $\fuss{X}$ this smallest $\reals$-semi-linear set.
We are interested in $\fuss{\orbit}$.

\begin{restatable}{theorem}{thmfusses}\label{thm:fusses}
The strongest $\reals$-semi-linear invariant $\fuss{\orbit}$ of $\orbit{}$ is computable.
\end{restatable}

Algebraic sets are those that are definable by finite unions and intersections of zeros of polynomials. For example, $\{(x,y)\mid xy=0\}$ describes the lines $x= 0$ and $y=0$. 
The (real) Zariski closure $\zar{X}$ of a set $X$ is the smallest algebraic subset of $\mathbb{R}^d$ containing the set $X$. The Zariski closure of the set of reachable points, $\zar{\orbit{}}$, can be computed algorithmically~\cite{HrushovskiOP018}.

An algebraic set $A$ is \emph{irreducible} if whenever $A\subseteq B \cup C$, where $B$ and $C$ are algebraic sets, then we have $A\subseteq B$ or $A\subseteq C$.  Any algebraic set (and in particular a Zariski closure) can be written effectively as a finite union of irreducible sets~\cite{chistov1986algorithm}.

\begin{proposition}\label{prop:fussxisaffirri} Let $\zar{X} = A_1\cup\dots\cup A_k$, with $A_i$'s irreducible.  Then $\fuss{X} = \fuss{\zar{X}} = \fuss{A_1} \cup \dots \cup  \fuss{A_k} = \aff{A_1} \cup \dots \cup  \aff{A_k}$.
\end{proposition}

\begin{proof}
Since $A_i \subseteq \fuss{X} = \cup_j L_j$, and $A_i$ is irreducible, we have $A_i \subseteq L_j$ for some $j$ (as the $L_j$'s are algebraic sets). Since $L_j$ is $\mathbb{R}$-linear, and $\aff{A_i}$ is the smallest $\mathbb{R}$-linear set covering $A_i$, we have $\aff{A_i} \subseteq L_j$. Taking $\fuss{X} = \aff{A_1} \cup \dots \cup  \aff{A_k}$ is thus optimal.
\qed \end{proof}

Thus $\fuss{\orbit}$ can be obtained by computing $\aff{A_i}$ for each irreducible $A_i$, where $\zar{\orbit} = A_1\cup\dots\cup A_k$.  To complete the proof of \cref{thm:fusses} it remains to confirm that affine hulls of algebraic sets can be computed algorithmically. Let us fix an algebraic set $A$, and let $W$ denote a set variable. Proceed as follows. Start with $W \leftarrow \set{x}$ for some point $x\in A$, and repeatedly let $W \leftarrow \aff{W \cup \set{y}}$, where $y \in A \setminus W$. Such a point $y$ can always be found using quantifier elimination in the theory of the reals. Each step necessarily increases the dimension, which can occur at most $d$ times, ensuring termination, at which point one has $\aff{A} = W$.

\section{Strongest \texorpdfstring{$\ints$}{Z}-linear Invariants}

Recall that a $\ints$-linear set $\set{q + p_1\ints + \dots + p_n \ints }$ is  defined by a base vector $q\in \ints^d$ and period vectors $p_1,\dots,p_n \in \ints^d$.  Equivalently, a $\ints$-linear set describes a \emph{lattice}, i.e., $\set{p_1\ints + \dots + p_n \ints }$, in $d$-dimensional space, translated to start from $q$ rather than $\vec{0}$.

\begin{theorem}\label{thm:smallestZlinear}Given a $d$-dimensional dynamical system $(\vecit{x}{0}, \set{M_1,\dots,M_k})$, 
  the \textit{strongest} $\ints$-linear inductive invariant containing the reachability set $\orbit{}$ exists
  and can be computed algorithmically.
\end{theorem}

The image of a $\ints$-linear set $L = \set{q + p_1\ints + \dots +  p_n\ints}$ by a matrix $M$ 
is the $\ints$-linear set: $M(L) = \set{ Mq + (Mp_1)\ints + \dots + (Mp_n)\ints }$. The following lemma asserts that when two points are in a $\ints$-linear set, the direction between these two points can be applied from any reachable point, and hence this direction can be included as a period without altering the set.

\begin{proposition}\label{prop:canadddir}
  Let $L = \set{q + a_1p_1+\dots+a_np_n\mid a_1,\dots,a_n\in\ints}$ be a $\ints$-linear set.
  If $x,y\in L$ then for all $z\in L$ and all $a'\in \ints$ we have $z + (y-x)a'\in L$.  In particular, we have
  $L = \set{q + a_1p_1+\dots+a_np_n + a'(y-x)\mid a_1,\dots,a_n,a'\in\ints}$.
\end{proposition}
\begin{proof}
If $x = q + a_1p_1+\dots+a_np_n$ and $y= q + b_1p_1+\dots+b_np_n$ then $y-x = q + b_1p_1+\dots+b_np_n - (q + a_1p_1+\dots+a_np_n)= (b_1-a_1)p_1+\dots+(b_n-a_n)p_n$.

Then for any $z = q + c_1p_1+\dots+c_np_n$, we have 
$z + a'(y-x) = q +  c_1p_1+\dots+c_np_n + a'((b_1-a_1)p_1+\dots+(b_n-a_n)p_n) = q +  (c_1 + a'(b_1-a_1))p_1+\dots+(c_n+ a'(b_n-a_n))p_n)$ where $ (c_i + a'(b_i-a_i))\in\ints$, so $z + a'(y-x) \in L$.
\qed \end{proof}

\begin{proposition}\label{prop:covering}
Given two $\ints$-linear sets $L_1=\set{q +p_1\ints + \dots +p_n \ints}$ and $L_2= \set{s + t_1\ints + \dots +t_m \ints}$, there exists a smallest $\ints$-linear set $L$ containing $L_1\cup L_2$: the set $L = \set{q + (s-q)\ints + p_1\ints + \dots + p_n \ints +  t_1\ints + \dots + t_m \ints }$.
\end{proposition}
\begin{proof}
First we show $L_1\cup L_2\subseteq L$:
\begin{itemize}
\item If $x = q + a_1p_1+\dots+a_np_n \in L_1$, then $ x= q + (s-q) 0+ a_1p_1+\dots+a_np_n + 0t_1+\dots+0t_m \in L$.

\item If $x = s +  b_1t_1+\dots+b_mt_m \in L_2$, then $x = q + (s-q)1 +  0p_1+\dots+0p_n +  b_1t_1+\dots+b_mt_m \in L$.
\end{itemize}
Next we show minimality as a straightforward consequence of~\cref{prop:canadddir}.

Clearly the vectors $p_1,\dots, p_n$ can be added by \cref{prop:canadddir} because any two points of $L_1$ differing by $p_i$ guarantees that adding $p_i$ does not alter the resulting set.
Similarly, $t_1,\dots, t_m$ can also be included.
Finally, by \cref{prop:canadddir}, the vector $s-q$ can be included because $q$ and $s$ both belong to $L_1\cup L_2$.
\qed \end{proof}

A $d$-dimensional lattice can always be defined by at most $d$ vectors; and thus if $d$ is the dimension of the matrices, no more than $d$ period vectors are needed in total. 
However, \cref{prop:covering} induces a representation which may over-specify the lattice by producing more than $d$ vectors to define the lattice. 
\begin{example}
  Consider the lattice  $\set{(2,2)\ints+ (0,6)\ints+(2,6)\ints }$, specified with three vectors, which is equivalent to the lattice $\set{(2,0)\ints + (0,2)\ints}$. Note that one may not simply pick 
an independent subset of the periods, as none of the following sets are equal:   $\set{(2,2)\ints + (0,6)\ints }$,  $\set{(2,2)\ints + (2,6)\ints }$,  $\set{(0,6)\ints + (2,6)\ints }$, and  $\set{(2,2)\ints+ (0,6)\ints+(2,6)\ints}$.
\end{example}

The \emph{Hermite normal form} can be used to obtain a basis of the vectors that define the lattice. Consider a lattice $L_i = \set{p_1\ints+\dots+p_d\ints}$. The lattice remains the same
if $p_i$ is swapped with $p_j$, if $p_i$ is replaced by $-p_i$, or if $p_i$ is replaced by $p_i + \alpha p_j$ where $\alpha$ is any fixed integer\footnote{The last replacement is valid, since if $x = y + \beta p_i \in L$ then  $x= y + \beta(p_i+\alpha p_j) -\beta\alpha p_j$ is in the new lattice.}. 

These are the unimodular operations. The Hermite normal form of a matrix $M$ is a matrix $H$ such that $M = UH$, where $U$ is a unimodular matrix (formed by unimodular column operations) and $H$ is lower triangular, non-negative and each row has a unique maximum entry which is on the main diagonal. Such a form always exists, and the columns of $H$  form a basis of the same lattice as the columns of $M$, because they differ up to unimodular (lattice-preserving) operations. There are many texts on the subject; we refer the reader to the lecture notes of Shmonin~\cite{ShmoninNotes2009} for more detailed explanations.

The columns of a matrix in Hermite normal form constitute a unique basis for the lattice (up to additional redundant zero columns). Hence a basis of minimal dimension can be obtained by computing the Hermite normal form of the matrix formed by placing the period vectors into columns.

We now prove the main theorem:
\begin{proof}[Proof of \cref{thm:smallestZlinear}] 
  We claim that \cref{lst:1} returns the strongest $\ints$-linear invariant $I$.

\cref{lst:1} proceeds in two phases:
\begin{itemize}
\item First find a necessary subset $L_0\subseteq I$ of the invariant having already the same dimension as $I$.
\item Then compute a growing sequence $L_0\subsetneq L_1\subsetneq \dots \subsetneq L_{m-1} =
  L_{m}= I$, where at each step the algorithm merely increases the density of the attendant sets in order to `fill in' missing points of the invariant.
\end{itemize}

Recall the set $R_0 = \set{\vecit{x}{0}, r_1,\dots,r_{d'}} \subseteq \orbit{}$, with $d' \le d$, from \cref{thm:tzeng}.  The resulting $\ints$-linear set $L_0 = \set{\vecit{x}{0} + (r_1 - \vecit{x}{0})\ints + \dots + (r_{d'} - \vecit{x}{0})\ints}$ is then a $d'$-dimensional porous subset of the $d'$-dimensional affine hull of the orbit ($L_0 \subseteq \aff{\orbit{}}$). Applying $M_1,\dots, M_k$ can only increase the density, but not the dimension. As each $r_i$ and $\vecit{x}{0}$ are in $\orbit{}$, by \cref{prop:canadddir} we can assume that each of the directions $(r_i - \vecit{x}{0})$ must be represented in any $\ints$-linear set containing $\orbit{}$, and we therefore have that $L_0 \subseteq I$.

In the second phase, we `fill in' the lattice as required to cover the whole of $\orbit$. To do this we repeatedly apply the covering procedure of \cref{prop:covering}. That is, $L_{i+1}$ is the smallest $\ints$-linear set covering $L_i \cup M_1(L_i)\cup\dots \cup M_k(L_i)$. To keep the number of vectors small, we keep the period vectors of the $\ints$-linear set in Hermite normal form.

The vectors $p_1 = (r_1 - \vecit{x}{0}), \dots,p_{d'} = (r_{d'} - \vecit{x}{0})$ form a parallelepiped (hyper-parallelogram) that repeats regularly. There are a finite number of integral points inside this parallelepiped. If new points are added in some step, they are added to every parallelepiped. Thus we can add new points finitely many times before saturating or becoming fixed. The volume of the parallelepiped is bounded above by $|p_1| \cdots |p_{d'}|$.

At each step, the volume of the parallelepiped must at least halve, thus the volume at step $t$ is $vol_t \le |p_1| \cdots |p_{d'}| / 2^t$. %
The procedure must saturate at or before the volume becomes $1$, which occurs after at most $\log(|p_1| \cdots |p_{d'}|) = \sum_i \log(|p_i|)$ steps. At each step, for efficiency considerations, we convert the $\ints$-linear set into Hermite normal form to retain exactly $d'$ period vectors. 

\begin{claim}[$I$ is the strongest invariant]
For every invariant $J$, we have $I\subseteq J$.
\end{claim}
By induction, let us prove that every invariant $J$ must contain $L_i$. Clearly this is the case for $L_0$ because all points of $R_0 \subseteq \orbit{}$ must be in $J$ and every period vectors in $L_0$ can be present, without loss of generality, thanks to \cref{prop:canadddir}. Assume $L_i\subseteq J$. Then it must be the case that $J$ contains every $M_j(L_i)$, as otherwise it would not be an invariant. It therefore follows that $J$ must contain $L_{i+1}$, since the latter is the minimal $\ints$-linear set containing $L_i$ and $M_j(L_i)$ for all $j\leq k$. Finally, since $I$ is itself one of the $L_i$'s, we have $I \subseteq J$ as required. \qed
\end{proof}

\begin{algorithm}[t]
\KwIn{$\vecit{x}{0}$,$M_1,\dots, M_k$ }

Compute $R_0 = \set{\vecit{x}{0},r_1,\dots, r_{d'}} \subseteq \orbit{}$

Compute $p_i = r_i - \vecit{x}{0}$ for $i\in \set{1,\dots,d'}$

$L_0 = \set{\vecit{x}{0}+ p_1\ints  +\dots  + p_{d'}\ints}$

\While{True} {

	$L_i = \operatorname{Covering}(L_{i-1}\cup M_1( L_{i-1})\cup \dots\cup M_k (L_{i-1}))$

	$H_i= \operatorname{HermiteNormalForm}(L_i)$

	$L_i = \set{\vecit{x}{0} + h_1\ints +\dots  + h_{d'}\ints \mid h_j 
\text{ column of } H_i}$
	
	\If{ $L_i = L_{i-1}$}{
		\Return{$L_i$}
	}
}
\caption{Strongest $\ints$-linear  invariant for LDS $(\vecit{x}{0},M_1,\dots,M_k)$}
\label{lst:1}
\end{algorithm}

\begin{remark}
Note that a $\ints$-linear set is not sufficient for the MU puzzle: both $1$ and $2$ are in the reachability set, thus $\set{1 + 1\ints}= \ints$ is the strongest $\ints$-linear invariant.
\end{remark}

\subsection{Extensions of \texorpdfstring{$\ints$}{Z}-linear sets without strongest invariants}\label{sec:no-strongest}

In this section we show that several generalisations of $\ints$-linear domains fail to admit strongest invariants. 

$\ints$-semi-linear sets are unions of $\ints$-linear sets, and therefore can include singletons. Consider the deterministic dynamical system starting from point $1$ and doubling at each step $\mathcal{M} = (1,(x\mapsto 2x))$. This system has reachability set $\orbit = \set{2^k \mid k\in \naturals}$, which is not even $\naturals{}$-semi-linear (our most general class). For this LDS we can construct the invariant  $\set{2,4,8,..., 2^k} \cup \set{2^{k+1} p_1 \mid p_1\in \ints }$ for each $k$. For any proposed strongest $\ints$-semi-linear invariant, one can find a $k$ for which the corresponding invariant is an improvement.

$\naturals{}$-linear sets generalise $\ints$-linear sets (observe that $\ints$-linear sets are a proper subclass, since $\set{x+ p_i\ints}$ can be expressed as $\set{x+ (-p_i)\naturals{} + p_i\naturals}$, but $\set{x+ p_i\naturals}$ is clearly not $\ints$-linear). Consider the LDS $((x_1,x_2),\left(\begin{smallmatrix}0&1\\1&0\end{smallmatrix}\right))$, with a reachability set consisting of just two points $x= (x_1,x_2)$ and $y=(x_2,x_1)$. There are two incomparable candidates for the minimal $\naturals{}$-linear invariant: $\set{x + (y-x)\naturals}$ and $\set{y + (x-y)\naturals}$. Similarly for $\rp$-linear invariants, the sets $\set{y + (x-y)\rp}$ and $\set{x + (y-x)\rp}$ are incomparable half-lines.

\subsection{\texorpdfstring{$\ints$}{Z}-linear targets}

We have so far only considered invariants for point targets. We now turn to lattice-like targets, in particular targets specified as \textit{full-dimensional} $\ints$-linear sets.

\begin{theorem}\label{thm:ztargets}
It is decidable whether a given LDS $(\vecit{x}{0},\set{M_1,\dots,M_k})$ reaches a \textit{full-dimensional} $\ints$-linear target $Y = \set{x + p_1\ints +\dots + p_d\ints}$, with $x,p_i\in\ints^d$.

Furthermore, for unreachable instances, a $\ints$-semi-linear inductive invariant can be provided.
\end{theorem}

\cref{thm:ztargets} requires the targets to be \textit{full-dimensional}. For nondeterministic systems reachability is undecidable for non-full-dimensional targets (in particular point targets)~\cite{markov1947certain}. However, even for deterministic systems, when $\ints$-linear targets fail to be \textit{full-dimensional} the reachability problem becomes as hard as the Skolem problem (see, e.g.~\cite{ouaknine2012decision}), for example by choosing as target the set $\set{(0,x_2,\dots,x_d) \mid x_2,\dots,x_d \in \mathbb{Z}}= \set{\vec{0} +e_2\ints +\dots +  e_d\ints}$, where $e_i \in \set{0,1}^d$ is the standard basis vector, with $(e_i)_i = 1$ and $(e_i)_j= 0$ for $i\ne j$.

Towards proving \cref{thm:ztargets}, we first show that \textit{full-dimensional} linear sets can be expressed as `square' hybrid-linear sets. Hybrid-linear sets are semi-linear sets in which all the components share the same period vectors, and thus differ only in starting position (whereas semi-linear sets allow each component to have distinct period vectors). By square, we mean that all period vectors are the same multiple of standard basis vectors.

\begin{lemma}\label{lemma:hybridifyz}
Let $Y = \set{x + p_1\ints + \dots + p_d\ints}$ be a \textit{full-dimensional} $\ints$-linear set.
Then there exists $m\in\naturals{}$ and a finite set $B\subseteq [0,m-1]^d$ such that $Y =$\linebreak $\bigcup_{b\in B} \set{b + me_1\ints +\dots +me_d\ints}$.
\end{lemma}
\begin{proof}
Suppose $p_1,\dots,p_d$ span a $d$-dimensional vector space. Let $P = \left(\begin{smallmatrix} p_1\\\vdots\\p_d\end{smallmatrix}\right)$ be the matrix with rows $p_1,\dots,p_d$. Since $P$ is full row rank  it is invertible, hence there exists a rational matrix $P^{-1}$ such that $e_i = P^{-1}_{i,1}p_1+\dots +P^{-1}_{i,d}p_d$. In particular let $m_i$ be such that $P^{-1}_{i,j}m_i$ is integral for all $j$. Then there is an integral combination of $p_1,\dots,p_d$ such that $m_ie_i$ is an admissible direction in $Y$.

Let $m = \lcm\set{m_1,\dots,m_d}$. Then $me_i$ is an admissible direction in $Y$.
Hence by \cref{prop:canadddir}, $Y$ is equivalent to $\set{x + p_1\ints + \dots + p_d\ints + me_1\ints+\dots+me_d\ints}$.
By the presence of $me_1\ints+\dots+me_d\ints$ we have that $x\in Y$ if and only $x' \in Y$ where $x'_i = (x_i \mod m)$.

And therefore $Y$ can be written as $\bigcup_{b\in B}\set{b + me_1\ints+\dots+me_d\ints}$, where
$B = [0,m-1]^d \cap Y$.
\qed
\end{proof}

We now prove \cref{thm:ztargets}.
\begin{proof}[Proof of \cref{thm:ztargets}]
Choose $m$ and $B$ as in \cref{lemma:hybridifyz}, so that $Y$ is of the form $\bigcup_{b\in B} \set{b + me_1\ints +\dots +me_d\ints}$.  We build an invariant $I$ of the form $\bigcup_{b\in B'} \set{b + me_1\ints +\dots +me_d\ints}$
for some $B'\subseteq [0,m-1]^d$.

We initialise the set $I_0=\set{x + me_1\ints +\dots +me_d\ints}$, where $x\in [0,m-1]^d$ such
that $x_j = (\vecit{x}{0}_j \mod m)$. 
We then build the set $I_1$ by adding to $I_0$ the sets $\set{y + me_1\ints +\dots +me_d\ints}$ 
where for each choice of $M_i$, $y\in[0,m-1]^d$ is formed by $y_j =( (M_ix)_j \mod m)$ for some 
$x\in I_0$. We iterate this construction until it stabilises in an inductive invariant $I$. Termination follows from the finiteness of $[0,m-1]^d$ (noting in particular that if termination occurs with $B' = [0,m-1]^d$, then $I = \ints^d$ which is indeed an inductive invariant).

If there exists $y\in B\cap I$
then return \textsc{Reachable}. This is because the same sequence of matrices applied to 
$\vecit{x}{0}$ to produce $y\in I$ would, thanks to the modulo step, wind up inside the set 
$\set{y + me_1\ints +\dots +me_d\ints}$, which is a part of the target.

Otherwise, return \textsc{Unreachable} and $I$ as invariant. By construction, $I$ is indeed an inductive invariant disjoint from the target set. \qed
\end{proof}

\begin{remark}
By the same argument, \cref{thm:ztargets} extends to a restricted class of $\ints$-semi-linear targets: the finite union of \textit{full-dimensional} $\ints$-linear sets.
\end{remark}

\section{\texorpdfstring{$\naturals$}{N}-semi-linear Invariants}
We now consider $\naturals{}$-semi-linear invariants, our most general class.  $\N$-semi-linear invariants gain expressivity thanks to the `directions' provided by the period vectors. For example, the only possible $\ints$-semi-linear invariant for the LDS $(0,(x\mapsto x+1))$ is $\ints$, yet the reachability set, $\N$, is captured exactly by an $\N$-linear  invariant. We show that a separating $\naturals$-semi-linear invariant can \textit{always} be found for unreachable instances of deterministic integer LDS, although the computed invariant will depend on the target. However, finding invariants is undecidable for nondeterministic systems, at least in high dimension. Nevertheless, we show decidability for the low-dimensional setting of the MU Puzzle---one dimension with affine updates.

\subsection{Existence of sufficient (but non-minimal) \texorpdfstring{$\naturals$}{N}-semi-linear invariants for point reachability in deterministic LDS}

Kannan and Lipton showed decidability of reachability of a point target for deterministic LDS~\cite{KannanL86}. In this subsection, we establish the following result to provide a separating invariant in unreachability instances.

\begin{theorem}
\label{thm:nsemi-linear}
Given a deterministic LDS $(x^{(0)},M)$ together with a point target $y$, if the target is unreachable then a separating $\N$-semi-linear inductive invariant can be provided.
\end{theorem}

To do so, we will invoke the results from~\cite{fijalkow2019monniaux} to compute an
$\mathbb{R}_+$-semi-linear inductive invariant, and then extract from it
an $\N$-semi-linear inductive invariant.
More precisely, the authors of~\cite{fijalkow2019monniaux} show how to build polytopic inductive invariants for certain deterministic LDS\@.  Such polytopes are either bounded or
are $\rp$-semi-linear sets. In the first case, the polytope contains only finitely many
integral points, which can directly be represented via an $\N$-semi-linear set. In the second case, 
we build an $\N$-semi-linear set containing exactly the set of integral points included in the 
$\rp$-semi-linear invariant, thanks to the following lemma.

\begin{restatable}{lemma}{rton}
\label{lem:RtoN}
Given an $\mathbb{R}_+$-linear set $S= \set{x + \sum_{i} p_i\mathbb{R}_+}$, 
where the vectors $p_i$ have rational coefficients and $x$ is an integer vector, 
one can build an $\N$-semi-linear set $N$ comprising precisely all of the integral 
points of $S$.
\end{restatable}

\begin{proof}[Proof of Theorem~\ref{thm:nsemi-linear}]
We note that every invariant produced in~\cite{fijalkow2019monniaux} has rational
period vectors, as the vectors are given by the difference of successive point in the orbit of the system, 
and thus Lemma~\ref{lem:RtoN} can be applied.
The authors of~\cite{fijalkow2019monniaux} build an inductive invariant in all cases except those for which every eigenvalue of the matrix governing the evolution of the LDS is either 0 or of modulus 1 and at least one of the latter is not a root of unity.  This situation however cannot occur in our setting. Indeed, the eigenvalues of an integer matrix are algebraic integers, and an old result of Kronecker~\cite{Kro57} asserts that unless all of the eigenvalues are roots of unity, one of them must have modulus strictly greater than $1$ (the case in which \emph{all} eigenvalues are $0$ being of course trivial).

This concludes the proof of Theorem~\ref{thm:nsemi-linear}.\qed
\end{proof}

\subsection{Undecidability of \texorpdfstring{$\naturals$}{N}-semi-linear invariants for nondeterministic LDS}

If the enhanced expressivity of $\N$-semi-linear sets allows us always to find an invariant for deterministic LDS, it contributes in turn to making the invariant-synthesis problem undecidable when the LDS is not deterministic. We establish this %
through a reduction from the infinite Post correspondence problem ($\omega$-PCP) that can be defined in the following way: given $m$ pairs of non-empty words $\{(u^{1},v^{1}),\dots, (u^{m},v^{m})\}$ over alphabet $\{0,2\}$, does there exist an infinite word $w=w_1w_2\dots$ over alphabet $\{1,\dots,m\}$ such that $u^{w_1}u^{w_2}\ldots = v^{w_1}v^{w_2}\ldots$.  This problem is known to be undecidable when $m$ is at least $8$~\cite{HH06,DL12}.

\begin{restatable}{theorem}{thmundec}\label{thm:undec}
The invariant synthesis problem for $\mathbb{N}$-semi-linear sets and linear dynamical systems with at least two matrices of size $91$ is undecidable.
\end{restatable}
\begin{proof}[Sketch]
We first establish the result in the case of several matrices in low dimension; this can then be transformed in a standard way to two larger matrices (of size $91$).

The proof is by reduction from the infinite Post correspondence problem.  Given an instance of this problem the pair of words corresponding to each sequence of tiles has an integer representation, using  base-4 encoding.
An important property of our encoding is that the operation of appending a new tile to an existing pair of words can be encoded by matrix multiplication.

Recall that if the instance of $\omega$-PCP is negative, then every generated pair of words will differ at some point.  Our encoding is such that this difference of letters creates a difference in their numerical encodings that can be identified with an $\mathbb{N}$-semi-linear invariant.  On the other hand, when there is a positive answer to the $\omega$-PCP instance, there can be no $\mathbb{N}$-semi-linear invariant.\qed \end{proof}

\subsection{Nondeterministic one-dimensional affine updates}
The previous section shows that point reachability for  nondeterministic LDS is undecidable once there sufficiently many dimensions, motivating an analysis at lower dimensions. The MU Puzzle requires a single dimension with affine updates (or equivalently two dimensions in matrix representation,
with the coordinate along the second dimension kept constant). We consider this one-dimensional affine-update case, and therefore, rather than taking matrices as input, we directly work with
affine functions of the form $f_i(x) = a_i x + b_i$.

\begin{theorem}\label{thm:1daffine}
Given $\vecit{x}{0},\target{}\in \ints$, along with a finite set of functions $\set{f_1,\dots,f_k}$ where $f_i(x) = a_ix+b_i$, $a_i,b_i \in\ints$ for $1 \le i \le k$, 
it is decidable whether $\target$ is reachable from $\vecit{x}{0}$. 

Moreover, when $y$ is unreachable, an $\naturals$-semi-linear separating inductive invariant can be algorithmically computed.
\end{theorem}

We note that decidability of reachability is already known~\cite{FinkelGH13,Fremont2013}. We refine this result by exhibiting an invariant which can be used to disprove reachability. In fact our procedure will produce an $\mathbb{N}$-semi-linear set which can be used to decide reachability, and which, in instances of non-reachability, will be a separating inductive invariant. We have implemented this algorithm into our tool \porous, enabling us to efficiently tackle the MU Puzzle as well as its generalisation to arbitrary collections of one-dimensional affine functions. We report on our experiments  in \cref{sec:tool}.

We build a case distinction depending on the type of functions that appear:
\begin{definition} A function $f(x) = ax + b$...
\begin{itemize}
\item ... is \emph{redundant} if $f(x)= b$, (including possibly $b = 0$), or if $f(x) = x$.
\item ... is \emph{counter-like} if $f(x) =  x + b$, $b\ne 0$. Two counter-like functions, $f(x) = x + b$ and $g(x) = x + c$ are \emph{opposing} if $b > 0$ and $c < 0$ (or vice-versa).
\item ... is \emph{growing} if $f(x) =  ax + b$ and $|a| \ge 2$. 
We say a growing function is \emph{inverting} if $a \le -2$.
\item ... is \emph{pure inverting} if $f(x) =  -x + b$.
\end{itemize}

\end{definition}

\subsubsection{Simplifying assumptions}
\begin{lemma}
  Without loss of generality, redundant functions are redundant; more precisely, we can reduce the computation of an invariant for a system having redundant functions to finitely many invariant computations for systems devoid of such functions.
\end{lemma}
\begin{proof} Clearly the identity function has no impact on the reachability set, and so can be removed outright. For any other redundant function, its impact on the reachability set does not depend on when the function is used, and we may therefore assume that it was used in the first step, or equivalently, using an alternative starting point. Hence the invariant-computation problem can be reduced to finitely many instances of the problem over different starting points, with redundant functions removed.
Finally, taking the union of the resulting invariants yields an invariant for the original system.
\qed \end{proof}

\begin{lemma}
Without loss of generality, $\vecit{x}{0} \ge 0$.
\end{lemma}
\begin{proof}
We construct a new system, where each transition $f(x) = ax +b$ is replaced by $\overline{f}(x) = ax - b$. Then $\vecit{x}{0}$ reaches $\target{}$ in the original system if and only if $-\vecit{x}{0}$ reaches $-\target$ in the new system. To see this, observe that if $f(x) = ax + b$, then $\overline{f}(-x) = -ax-b = - f(x)$.
\qed \end{proof}

\begin{lemma}\label{wlog:twopureinvertertwooppcounters}
  Suppose there are at least two distinct pure inverting functions (and possibly other types of functions). Then without loss of generality there are two opposing counters. 
\end{lemma}
\begin{proof}
Consider $f(x) = -x + b$, and $g(x) = -x + c$. Then $f(g(x)) = -(-x+c)+b = x +b-c$ and $g(f(x)) = -(-x+b)+c = x+c-b$. Since $b-c = -(c-b)$ and $b\neq c$ (as $f\neq g$) these two functions are opposing. 
\qed \end{proof}

\subsubsection{Two opposing counters.}

Let us first observe that when there are two opposing counters, we essentially move in either direction by some fixed amount. This will entail that only $\ints$-(semi)-linear invariants can be produced, rather than proper $\naturals{}$-(semi)-linear invariants.
\begin{restatable}{lemma}{twooppcounter}\label{lemma:twooppcounter}
Suppose there are two opposing counters, $f(x) = x + b$, and $g(x) = x - c$. Then for any reachable $x$ we have $\set{x + d\ints}\subseteq I$ for $d = \gcd(b,c)$.
\end{restatable}

Therefore, starting with $\set{\vecit{x}{0} + d\ints} \in I$ we can `saturate' the invariant under construction using the following lemma:
\begin{lemma}\label{lemma:zsaturation}
Let $h(x) = x + d$ be chosen as a reference counter amongst the counters.  If $\set{x + d\ints} \in I$,  then $\set{f(x) + d\ints}\in I$ for every function $f$.
\end{lemma}
\begin{proof}[Proof of \cref{lemma:zsaturation}]
  Consider the function $f(x) = ax + b$. If $x = y+dk \in I$, then $f(x)=ax+b
= ay + adk + b = f(y) + adk \in I$.

Now thanks to the presence of counter $h(x) = x+d$, by choosing the initial $k \in\ints$ appropriately and applying $h(x)$ sufficiently many times (say $m \in \naturals$ times), one can reach $f(x) + adk  + dm = f(x) + dn$ for any desired $n\in\ints$.
\qed \end{proof}

Without loss of generality if $\set{x +d\ints}$ is in the invariant, then $0\le x < d$. We then repeatedly use \cref{lemma:zsaturation} to find the required elements of the invariant. Since there are only finitely many residue classes (modulo $d$), every reachable residue class  $\set{c_1,\dots,c_n}$ can be found by saturation (in at most $d$ steps), yielding invariant $\set{c_1+d\ints}\cup \dots \cup \set{c_n+d\ints}$.

Thanks to \cref{wlog:twopureinvertertwooppcounters},  in all remaining cases there is without loss of generality at most one pure inverter.

\subsubsection{Only pure inverters.}
If there is exactly one pure inverter $f(x) = -x + b$ (and no other types of functions), then $f(\vecit{x}{0}) = -\vecit{x}{0} + b$ and $f(-\vecit{x}{0} + b)= \vecit{x}{0}-b + b = \vecit{x}{0}$, thus the reachability set is finite, with exact invariant $\set{\vecit{x}{0}, -\vecit{x}{0}+b }$.

\subsubsection{No Counters.}
If we are not in the preceding case and there are no counters, then there must be growing functions and by \cref{wlog:twopureinvertertwooppcounters}, without loss of generality at most one pure inverter. We show that all growing functions increase the modulus outside of some bounded region.

\begin{lemma}\label{lemma:growthbound}
For every $M \ge 0$ and every growing function $f(x) = ax + b$, $|a| \ge 2$, there exists $C^M_f \ge 0$ such that if $|x| \ge C^M_f$ then $|f(x)| \ge |x| + M$.
\end{lemma}
\begin{proof} By the triangle inequality we have: $|f(x)| = |ax+b| \ge |a||x| - |b|$. Thus $|x| \ge \frac{|b|+|M|}{|a| -1} \implies |a||x| - |b| \ge |x| + |M| \implies |f(x)|\ge |x| +M$.\qed
\end{proof}

This is the only situation in which the invariant is not exactly the reachability set, and requires us to take an overapproximation.

Let $C = \max\set{C^0_{f_1},\dots,C^0_{f_k}, |\target|+1}$, for $f_1,\dots,f_k$ growing functions. If there are no pure inverters then $\set{-C - \naturals} \cup \set{C + \naturals}$ is invariant (although may not yet contain the whole of $\orbit$). However, we can return the inductive invariant $\set{-C - \naturals} \cup \set{C + \naturals} \cup (\orbit{} \cap (-C, C))$. The set $\orbit{} \cap (-C, C)$ is finite and can elicited by exhaustive search, noting that once an element of the orbit reaches absolute value at least $C$, the remainder of the corresponding trajectory remains forever outside of $(-C,C)$.

If there is one pure inverter $g(x) = -x + d$ then observe that $-C$ is mapped to $C +d$ and $C+d$ is mapped to $-C$. Thus intuitively we want to use the interval $(-C, C+d)$. However two problems may occur: (a) since $d$ could be less than $0$ then $C+d$ may no longer be growing (under the application of the growing functions), and (b) an inverting growing function only ensures that $-C$ is mapped to a value greater than or equal to $C$, rather than $C+d$. Hence, we choose $C'$ to ensure that $C'\pm d$ is still growing by at least $|d|$ (under the application of our growing functions). Let $C' = \max\set{C^{|d|}_{f_1},\dots,C^{|d|}_{f_k}, |\target|+1}+|d|$. Then the invariant is $\set{-C' - \naturals} \cup \set{C'+d + \naturals} \cup (\orbit{} \cap (-C', C'+d))$.

\subsubsection{Non-opposing counters.}
The only remaining possibility (if there do not exist two opposing counters, and not all functions are growing or pure inverters), is that there are counter-like functions, but they are all counting in the same direction. There may also be a single pure inverter, and possibly some growing functions. 

Pick a counter $h(x) = x + d$ to be the reference counter; the choice is arbitrary, but it is convenient to pick a counter with minimal $|d|$. As a starting point, we have $\set{\vecit{x}{0} + d\naturals} \subseteq I$.

\begin{lemma}\label{lemma:inverterzs}
If there is an inverter $g(x) = -ax + b$, with $a> 0,b\in\ints$, and we have $\set{x + d\naturals} \subseteq I$ then $\set{g(x) + d\ints}\subseteq I$.
 \end{lemma}
The crucial difference with \cref{lemma:zsaturation} is the observation that now an $\naturals$-linear set has induced a $\ints$-linear set.
\begin{proof}
Let $r =g(x) + dm$ for $m \in \ints$. We show $r \in I$. Consider $x + dn$ for $n\in\naturals$, then $g(x+dn) = -a(x+dn) + b = -ax+b-adn = g(x) -adn$. Hence $g(x) -adn + dk$, $n,k\in\naturals$, is reachable by applying $k$ times the function $h(x)$. Hence for any $m\in\ints$ there exists $k,n\in\naturals$ such that $k - na = m$, so that $r$ is indeed reachable.
\qed \end{proof}

Similarly to the situation with two opposing counters, whenever the invariant contains some $\ints$-linear set, \cref{lemma:zsaturation} allows us to saturate amongst the finitely many reachable residue classes.

However, the invariant may contain subsets that are not $\ints$-linear.
Consider $\set{x + d\naturals}\subseteq I$, which is not yet invariant. We repeatedly apply non-inverting functions to $\set{x + d\naturals}$ to obtain new $\naturals$-linear sets (not $\ints$-linear sets).  When the function applied `moves' in the direction of the counters this will ultimately saturate (in particular when applying other counter functions). However, in the opposite direction, we may generate infinitely many such classes.

\begin{example}
Consider the reference counter $h(x) = x +4$, with initial point $5$. This yields an initial set $\set{5 + 4 \naturals}\subseteq \orbit$, where $5$ is the initial point and $4\naturals{}$ is derived from the counter increment.  Now when applying $x \mapsto 2x +6$ to $\set{5 + 4 \naturals}$ we obtain $\set{10+ 6 + 8 \naturals  + 4\naturals}=\set{16+4\naturals}$, then $\set{38+4\naturals}$, and then $\set{82 + 4\naturals}$. However $\set{82+4\naturals} \subseteq\set{38+4\naturals}$ and we can therefore stop with the invariant $\set{5+4\naturals} \cup \set{16+4\naturals} \cup \set{38+4\naturals}$.

However, if the initial sequence is not moving in the direction of the reference counter, this saturation does not occur. Consider $\set{5 + 4 \naturals}$ with the function $x\mapsto 2x - 6$. Then $\set{5 + 4 \naturals}$ maps to $\set{10-6  + 8\naturals+ 4 \naturals}= \set{4+4\naturals}$, which maps to $\set{2 + 4 \naturals}$, $\set{-2 + 4 \naturals}$, $\set{-10 + 4 \naturals}$, $\set{-26 + 4 \naturals}$, and so on. However $-2$ and $-10$ are both $2$ modulo $4$ (and so is $-26$ as well). This means in the negative direction we can obtain arbitrarily  large negative values  congruent to $2$ modulo $4$ and then use the reference counter $h(x) = x + 4$ to obtain any value of $\set{2 + 4\ints}$.\qed
\end{example}

Clearly we can examine all reachable residue classes defined by our reference counter. Any residue class reachable after an inverting function induces a $\ints$-linear set. So it remains to consider those $\naturals{}$-linear sets reachable without inverting functions. The remaining case to handle occurs when we repeatedly induce $\naturals{}$-linear sets until they repeat a residue class in the direction opposite to that of the reference counter. 

We consider the case for $h(x) = x+ d$ with $d\ge 0$. The case with $h(x) =x-d$ is symmetric. It remains to detect when a set  $\set{x+d\naturals}$ leads to $\set{y+d\naturals}$ by a sequence of non-inverting functions with $x\equiv y\mod d$. Then by repeated application of these functions one can reach sets $\set{z+d\naturals}$ with $z$ arbitrarily small, hence we can replace  $\set{x+d\naturals}$ by $\set{x+d\ints}$. We give further details in the appendix.

\subsubsection{Reachability.}

The above procedure is sufficient to decide reachability. In all cases apart from that in which there are no counters, the invariants produced coincide precisely with the reachability sets. A reachability query therefore reduces to asking whether the target belongs to the invariant.

In the remaining case, the invariant obtained is parametrised by the target via the bound $C'$. The target lies within the region $(-C',C' + d)$, within which we can compute all reachable points. Thus once again, the target is reachable precisely if it belongs to the invariant. However, for a new target of larger modulus, a different invariant would need to be built.

\subsubsection{Complexity.}
\begin{restatable}{lemma}{complexityunarypoly}\label{lemma:complexityunarypoly}
Assume that all functions, starting point, and target point are given in unary. Then the invariant can be computed in polynomial time.
\end{restatable}

Without the unary assumption, the invariant could have exponential size, and hence require at least exponential time to compute. That is because the invariant we construct could include every value in an interval, for example, $(-C,C)$, where $C$ is of size polynomial in the largest value.

As shown in~\cite{Fremont2013}, the reachability problem is at least $\NP$-hard in binary, because one can encode the integer Knapsack problem (which allows an object to be picked multiple times rather at most once). Moreover the Knapsack problem is efficiently solvable in pseudo-polynomial time via dynamic programming; that is, polynomial time assuming the input is in unary, matching the complexity of our procedure.

\section{The POROUS Tool}\label{sec:tool}
Our invariant-synthesis tool \porous\footnote{Tool: \href{http://invariants.davidpurser.net}{\color{blue}invariants.davidpurser.net} Code: \href{https://github.com/davidjpurser/porous-tool}{\color{blue}github.com/davidjpurser/porous-tool}} computes $\naturals{}$-semi-linear invariants for  point and $\mathbb{Z}$-linear targets on systems defined by  one-dimensional affine functions. \porous\ includes implementations of the
  procedures of \cref{thm:ztargets} (restricted to one-dimensional affine systems) and \cref{thm:1daffine}. \porous\ is built in Python and can be used by command-line file input, a web interface, or by directly invoking the Python packages.

  \porous\ takes as input an instance (a start point, a target, and a collection of functions) and returns the generated invariant. Additionally it provides a proof that this set is indeed an inductive invariant: the invariant is a union of $\naturals{}$-linear sets, so for each linear set and each function, \porous\ illustrates the application of that function to the linear set and shows for which other linear set in the invariant this is a subset. Using this invariant, \porous\ can decide reachability; if the specific target is reachable the invariant is not in itself a proof of reachability (since the invariant will often be an overapproximation of the global reachability set).
Rather, equipped with the guarantee of reachability, \porous\ searches for a direct proof of reachability: a sequence of functions from start to target (a process which would not otherwise be guaranteed to terminate).

\subsubsection{Experimentation.}
\porous\ was tested on all $2^7-1$ possible combinations of the following function types, with $a\ge2, b\ge1$: positive counters ($x\mapsto x+b$), negative counters ($x\mapsto x-b$), growing ($x\mapsto ax\pm b$), inverting and growing ($x\mapsto -ax\pm b$), inverters with positive counters ($x\mapsto -x+b$), inverters with negative counters ($x\mapsto -x-b$) and the pure inverter ($x\mapsto -x$). For each such combination a random instance was generated, with a size parameter to control the  maximum modulus of  $a$ and $b$, ranging between 8 and 1024. The starting point was between 1 and the size parameter and the target was between 1 and 4 times the size parameter. Ten instances were tested for each size parameter  and each of the $2^7-1$ combinations, with between 1 and 9 functions of each type (with a bias for one of each function type).

\begin{table}[t]
\begin{adjustbox}{width=\textwidth,center}
\noindent\begin{tabular}{|p{0.6cm}|p{0.8cm}|p{0.9cm}||p{1.9cm}|p{0.8cm}|p{0.9cm}||p{1.85cm}|p{1.85cm}|p{1.7cm}|}
\hline
Size & \multicolumn{2}{p{1.6cm}||}{ Invariant Build Time} & Unreachable Instances & \multicolumn{2}{p{1.7cm}||}{Invariant Proof Time} & Reachable Instances & Reachable with proofs & Reachability Proof time \\
 & avg & max & & avg & max &  & within ${\approx}30$s & avg  \\\hline
8 & 0.001 & 0.009 &  100 (9.84\%) & 0.005 & 0.261 &  916 (90.2\%) &  911 (99.5\%) & 0.033 \\ 
16 & 0.001 & 0.020 &  122 (12.0\%) & 0.010 & 0.788 &  894 (88.0\%) &  885 (99.0\%) & 0.053 \\ 
32 & 0.003 & 0.068 &  134 (13.2\%) & 0.020 & 0.911 &  882 (86.8\%) &  843 (95.6\%) & 0.203 \\ 
64 & 0.008 & 0.261 &  150 (14.8\%) & 0.052 & 2.969 &  866 (85.2\%) &  766 (88.5\%) & 0.294 \\ 
128 & 0.021 & 0.557 &  153 (15.1\%) & 0.096 & 2.426 &  863 (84.9\%) &  719 (83.3\%) & 0.464 \\ 
256 & 0.088 & 2.838 &  166 (16.3\%) & 0.316 & 43.587 &  850 (83.7\%) &  620 (72.9\%) & 0.998 \\ 
512 & 0.428 & 9.312 &  162 (15.9\%) & 0.899 & 21.127 &  854 (84.1\%) &  570 (66.7\%) & 1.120 \\ 
1024 & 1.121 & 20.252 &  173 (17.0\%) & 3.275 & 65.397 &  843 (83.0\%) &  514 (61.0\%) & 1.646 \\ \hline
all & 0.209 & 20.252 &  1160 (14.3\%) & 0.584 & 65.397 &  6968 (85.7\%) &  5828 (83.6\%) & 0.499 \\
\hline
\end{tabular}
\end{adjustbox}
\caption{Results varying by size parameter (last row includes all instances tested). Times  are given in seconds, with the average and maximum shown (except reachability proof time, which are all approximately 30s due to instances that terminate just before the timeout).}
\label{fig:results}
\end{table}

Our analysis, summarised in \cref{fig:results}, illustrates the effect of the size parameter. The time to produce the proof of invariant is separated from the process of building the invariant, since producing the proof of invariant can become slower as $|I|$ becomes larger; it requires finding $L_k\in I$  such that $f_i(L_j) \subseteq L_k$ for every linear set $L_j\in I$ and every affine function $f_i$. In every case \porous\ successfully built the invariant, and hence decided reachability very quickly (on average well below 1 second) and also produced the proof of invariance in around half a second on average. To demonstrate correctness in instances for which the target is reachable \porous\ also attempts to produce a proof of reachability (a sequence of functions from start to target). Since our paper is focused on invariants as certificates of non-reachability, our proof-of-reachability procedure was implemented crudely as a simple breadth-first search without any heuristics, and hence a timeout of 30 seconds was used for this part of the experiment only.

Our experimental methodology was partially limited due to the high prevalence of reachable instances. A random instance will likely exhibit a large (often universal) reachability set. 
When two random counters are included, the chance that $\gcd(b_1,b_2) = 1$ (whence the whole space is covered) is around $60.8\%$ and higher if more counters are chosen.

Overall around 86\% of instances were reachable (of which 84\% produced a proof within 30 seconds).   Of the 14\% of unreachable instances, all produced a proof, with the invariant taking around 0.2 seconds to build and 0.6 seconds to produce the proof. The  30-second timeout when demonstrating reachability directly is several orders of magnitudes longer than answering the reachability query via our invariant-building method.

A typical academic/consumer laptop was used to conduct the timing and analysis  (a four-year-old, four-core MacBook Pro).

\section{Conclusions and Open Directions}
We introduced the notion of porous invariants, which are not necessarily convex and can in fact exhibit infinitely many `holes', and studied these in the context of multipath (or branching/nondeterministic) affine loops over the integers, or equivalently nondeterministic integer linear dynamical systems. We have in particular focused on reachability questions. Clearly, the potential applicability of porous invariants to larger classes of systems (such as programs involving nested loops) or more complex specifications remains largely unexplored.

Our focus is on the boundary between decidability and undecidability, leaving precise complexity questions open. Indeed, the complexity of synthesising invariants could conceivably be quite high, except where we have highlighted polynomial-time results. On the other hand, the invariants produced should be easy to understand and manipulate, from both a human and machine perspective.

On a more technical level, in our setting the most general class of invariants that we consider are $\naturals{}$-semi-linear. There remains at present a large gap between decidability for one-dimensional affine functions, and undecidability for linear updates in dimension 91 and above. It would be interesting to investigate whether decidability can be extended further, for example to dimensions 2 and 3.

\subsubsection*{Acknowledgements}
\begin{sloppypar}
This work was funded by DFG grant 389792660 as part of TRR~248 (see \href{https://perspicuous-computing.science}{\color{blue}perspicuous-computing.science}). Jo\"el Ouaknine was 
supported by ERC grant AVS-ISS (648701), and is also
affiliated with Keble College, Oxford as \href{http://emmy.network/}{\color{blue}emmy.network} Fellow.
James Worrell was supported by EPSRC Fellowship EP/N008197/1.
\end{sloppypar}

\bibliographystyle{splncs04}

\newpage
\appendix

\section{Proof of \cref{thm:tzeng}}
\label{appen:thm:tzeng}
\thmtzeng*
\begin{proof}
The result of \cite[Theorem 3.1]{Tzeng92} proceeds by finding new points in the reachability set and adding them to a set of points if the new point is linearly independent of the current version this set. Whilst the result of  \cite{Tzeng92} refers to linear independence, this can be converted to affine independence by increasing the dimension by one.

The procedure works via a pruned version breadth-first search, with nodes only expanded if its children are linearly independent of the current set. Hence, the first point found in the tree is the initial point $\vecit{x}{0}$, and therefore this point is included. The maximum depth of the tree that needs to be explored is $d$, and so every point included is reached with at most $d$ applications of matrices to $\vecit{x}{0}$. Hence, if the largest absolute value of a point or matrix entry is $\mu$, after $d$ iterations, the largest absolute value is $d^{d-1}\mu^d$. This is by induction on the largest possible value $\mu$ for every entry:
\[\text{Base case: }
\begin{pmatrix}
\mu & \dots & \mu\\ 
& \ddots\\
\mu & \dots & \mu\\ 
\end{pmatrix}
\begin{pmatrix}
d^0\mu \\ 
\vdots\\
d^0\mu \\ 
\end{pmatrix}
=
\begin{pmatrix}
d\mu^2 \\ 
\vdots\\
d\mu^2 \\ 
\end{pmatrix}
\]
\[\text{Inductive  case: }
\begin{pmatrix}
\mu & \dots & \mu\\ 
& \ddots\\
\mu & \dots & \mu\\ 
\end{pmatrix}
\begin{pmatrix}
d^{k-1}\mu^k \\ 
\vdots\\
d^{k-1}\mu^k \\ 
\end{pmatrix}
=
\begin{pmatrix}
d\mu(d^{k-1}\mu^k) \\ 
\vdots\\
d\mu(d^{k-1}\mu^k) \\ 
\end{pmatrix}
=
\begin{pmatrix}
d^k\mu^{(k+1)}\\ 
\vdots\\
d^k\mu^{(k+1)}\\ 
\end{pmatrix}
\]

The result of \cite{Tzeng92} is in polynomial time in the number of arithmetic operations, we observe that this is also polynomial time in the bit-size. The independence checking in the algorithm involves checking linear independence  of at most $d$ vectors all having bit size at most $log((d\mu)^{d}) = d\log(d) + d\log(\mu)$, which can be done in polynomial time in the bit-size (for example by Bareiss algorithm for calculating the determinant).\qed
\end{proof}

\section{Proof of \cref{lem:RtoN}}
\rton*
\begin{proof}
Let $S= \set{x + \sum_{i} p_i\mathbb{R}_+}$ be a $\mathbb{R}_+$-linear set
where the vectors $p_i$ have rational coefficients and $x$ is an integer vector.
Let $k\in\N$ so that the vectors $kp_i$ have integer coefficients. We denote by 
$v_i$ the integer vectors obtained as a convex combination of the vectors $kp_i$.
Then the set $T = \set{x + \sum_{i} v_i\N}$ contains exactly the integer vectors contained
in $S$.

Indeed, first $T$ only contains integer points as both $x$ and the vectors $v_i$ are integer 
vectors. Secondly, all the vectors in $T$ are included in $S$ as the period vectors of $T$ lie
in the convex hull of the vectors of $S$. 
Finally, given an integer vector $y$ in $S$,
$y$ can be rewritten as $y = x + v + \sum_{i} m_ikp_i$ where for all $i, m_i\in \N$ and $v$ is an integer vector lying in 
the convex hull of the vectors $kp_i$. Therefore there exists $j$ such that $v_j=v$ and as
for all $i$, $kp_i$ is a period vector of $T$, $y\in T$.\qed
\end{proof}

\section{Proof of \cref{thm:undec}}

\thmundec*

\begin{proof}
We will prove the result for $m+5$ matrices of size $7$. This can then be transformed in a usual way to two matrices of size $7m+35$ (See Theorem~9 of~\cite{fijalkow2019monniaux} for instance).

In order to simplify the main part of the proof, let us first show that one can enforce an order
between the matrices using affine transformations on one dimension. Let us denote $p$ this 
dimension, it is initially equal to 1 and its target value is 0. 
Consider the three following affine transformation: $f_1(p)= 2p -1$, $f_2(p) = 2p -2$ and 
$f_3(p) = 2p$, then the only sequences of transformation allowing to reach the target are of 
the form $f_3^*f_2f_1^*$. Indeed, let $\I = \{p\mid p\geq 2 \vee p\leq -1\}$, 
we have (1) if $p\in \I$, then for all $i\in \{1,2,3\}$, 
$f_i(p)\in \I$,
(2) $f_1(1) = 1$ and $f_1(0)\in \I$,
(3) $f_2(1) = 0 $ and $f_2(0)\in \I$ and
(4) $f_3(1)\in \I$ and $f_3(0)=0$.
As a consequence, the inductive invariant $\I$ ensure that any sequence of transformation
that do not have the desired order cannot reach the target. In the following, we will call type 
1, 2 or 3 the transformations we define, depending on whether they implictly contain the function
$f_1$, $f_2$ or $f_3$.

We reduce an instance $\{(u^{1},v^{1}),\dots, (u^{m},v^{m})\}$ of the $\omega$-PCP problem 
to the invariant synthesis problem.
In order to simplify future notations, given a finite or infinite word $w$,
we denote by $|w|$ the length of the word $w$ and given an integer $i\leq |w|$,
we write $w_i$ for the $i$'th letter of $w$.
Given a finite or infinite word $w$ on alphabet $\{1,\dots,m\}$ we denote by $u^{w}$ and 
$v^{w}$ the words on the alphabet $\{0,2\}$ such that $u^{w}= u^{w_1}u^{w_2}\dots$ and 
$v^{w}= v^{w_1}v^{w_2}\dots$. 
Given a (finite or infinite) word $w$ on 
the alphabet $\{0,2\}^*$, denote by
$\enc{w}=\sum_{i=1}^{|w|}w_i4^{|w|-i}$ the quaternary encoding of $w$. 
It is
clear that it satisfies $\enc{ww'}=4^{|w'|}\enc{w}+\enc{w'}$.
For all $i\leq m$, we denote by ${n_i}=4^{|u^{i}|}$, $m_i=4^{|v^{i}|}$ 
and $\max_i = \max (n_i,m_i)$.

We work with 5 dimension, $(s,c,d,n,k)$, and define the following transformations:
\begin{itemize}
\item For $i \leq m$, the type 1 transformation $\mathsf{Simulate_i}$ on $(s,c,d,n,k)$ encode 
the action 
of reading the pair $(u^i,v^i)$ and increases the counters $n$ and $k$: it simultaneously applies
$s \leftarrow \max_i s + [u^i]\frac{\max_i}{n_i}-[v^i]\frac{\max_i}{n_i}$,
$c \leftarrow \frac{\max_i}{n_i}c$,
$d \leftarrow \frac{\max_i}{n_i}d$, $n\leftarrow n+k$ $k\leftarrow k+1$.
\item The type 2 transformation $\mathsf{Transfer}$ on $(s,c,d,n,k)$ gather some of the values 
in order to compare them:
$s\leftarrow s  - c - d$, $c \leftarrow -s -c-d$.
\item The type 3 transformation $\mathsf{Inc_s}$ increments $s$: $s\leftarrow s+1$.
\item The type 3 transformation $\mathsf{Inc_c}$ increments $u$: $c\leftarrow c+1$.
\item The type 3 transformation $\mathsf{Dec}$ decreases $k$ and $n$: $n\leftarrow n-k$, 
$k\leftarrow k-1$.
\item The type 3 transformation $\mathsf{Dec_k}$ decrements $k$: $k\leftarrow k-1$.
\end{itemize}
These $m+5$ transformations need 7 dimensions in total: the five above, $(s,c,d,n,k)$, the one used for ordering the transformations,$p$, and one dimension constantly equal to 1, 
required to use affine transformations.

We now show that there is a solution to the given instance of the $\omega$-PCP problem iff
there does not exist a $\mathbb{N}$-semi-linear invariant 
for the system with initial point $x = (0,1,1,0,0,1,1)$, target $y=(0,0,0,1,0,0,1)$
and using the matrices inducing the transformations defined above.

Assume first that there is a solution $w$ to the $\omega$-PCP instance.
Consider the sequence of points $(x_n)$ obtained as follows:
for all $j\in \mathbb{N}$, denoting $w_{\leq j}$ the prefix of $w$ of length, 
$x_j = (s_j, c_j,0, n_j,k_j,0,1) = \mathsf{Transfer}\ \mathsf{Simulate_{w_{\leq j}}} x$ 
where $\mathsf{Simulate_{w_{\leq j}}}$ represents the transformation 
$\mathsf{Simulate_{w_{j}}}\dots\mathsf{Simulate_{w_2}}\mathsf{Simulate_{w_1}}$.
We have that $s_j$ and $c_j$ are negative. Indeed, let $(s,c,d)$ be the three first components of 
$\mathsf{Simulate_{w_{\leq j}}} x$, we have that 
$s=c\enc{u^{wi}}-d\enc{v^{wi}}$. 
As $w_{\leq j}$ is a prefix of a solution to the 
$\omega$-PCP instance, assuming $|u^wi|\leq |v^wi|$ this implies that 
\begin{align*}
|s| & = |c\enc{u^{wi}}-d\enc{v^{wi}}|\\
& = \sum_{j=1}^{|v^wi|}|u^{wi}_j-v^{wi}_j|c4^{|u^wi|-j} \\
& = \sum_{j=|u^wi|+1}^{|v^wi|}v^{wi}_jc4^{|u^wi|-j} \\
& \leqslant \frac{2c}{3}
\end{align*}
Thus $|s|-c-d$ is negative, thus $s_j = s-c-d$ and $c_j = -s-c-d$ are negative.

Due to the above, by applying to the points $x_j$ a number of time the transformations
$\mathsf{Inc_s}$ and $\mathsf{Inc_c}$, we obtain the sequence of points $(y_j)$ where 
$y_j = (0,0,0,n_j,k_j,0,1)$.
We claim that any semi-linear invariant containing all the points $y_j$ also contains a point
of the shape $(0,0,0,0,n_j+d,k_j,0,1)$ where $d$ is a positive integer. 
This will imply the result as from such a point, one can 
reach the target by $d-1$ applications of $\mathsf{Dec_k}$ and $k_j$ applications of 
$\mathsf{Dec}$ and thus there is no semi-linear invariant 
of the system that does not intersect the target.

Let us now prove the above claim.
Let $\I$ be a semi-linear set containing every point $(y_j)$ (which we will see as two-dimensional objects by projecting on the 4th and 5th dimension). Then there exists a linear
set $\I'\subseteq \I$ that contains infinitely many vectors of $(y_j)$.
This set $\I'$ is defined by an initial vector, and a set of period vectors. 
As $\I'$ contains
infinitely many vectors of $(y_j)$ where the ratios between the first and second component 
is increasing, one of the period vectors is of the form
$(d,0)$ where $d$ is a strictly positive integer. 
Let $j$ be such that $y_j\in\I'$, then $(n_j+d, k_j)\in \I'$ which implies the claim. 

As a consequence, every $\mathbb{N}$-semi-linear set over-approximating the system intersects 
with the target.

Conversely, assume that there is no solution to the $\omega$-PCP instance.
There exists $n_0\in \N$ such that for every infinite word $w$ on alphabet
$\{0,\dots, m\}$ there exists $n\leq n_0$ such that $u^{w}_n \neq v^{w}_n$.
Indeed, consider the tree which root is labelled by $(\eps,\eps)$ and, given a node $(u,v)$ of the tree,
if for all $n\leq \min(|u|,|v|)$ we have $u_n=v_n$, then this node has $m$ children:
the nodes $(u u^{i},v v^{i})$ for $i=1\dots m$. This tree is finitely branching and
does not contain any infinite path (which would induce a solution to the $\omega$-PCP instance).
Thus, according to K\"onig's lemma, it is finite. We can therefore choose the height of this tree as our $n_0$.

We define the invariant $\I=\I_1\cup \I_2\cup \I_3$ where
\begin{align*}
\I_1=&\big\{\mathsf{Simulate_w} (x) \mid w\in \{1,\dots,m\}^* \wedge |w|\leq n_0+1\big\},\\
\end{align*}
\begin{align*}
\I_2=& \big\{z=(s,c,0,n,k,0,1) \mid  z= (\mathsf{Inc_s})^*(\mathsf{Inc_c})^*
(\mathsf{Dec})^*(\mathsf{Dec_k})^* \mathsf{Transfer}\  \mathsf{Simulate_w} (x) 
\\ &\wedge w\in \{1,\dots,m\}^* \wedge |w|\leq n_0+1 \wedge s,t,n,k\in \N \big\}
\end{align*} 
and 
\begin{align*}
\I_3= &\big\{(s,c,d,n,k,p,1)\mid (|s|-c-d\geq 1\wedge c\geq 0 \wedge d \geq 0 \wedge p=1)\\ 
&\vee ((s\geq 1\vee c\geq 1\vee n\leq -1 \vee k\leq -1) \wedge p=0) \vee p\leq -1 \vee p \geq 2\}.
\big\}
\end{align*}

By definition, $\I$ is an $\mathbb{N}$-semi-linear set, contains $x$ and does not contain $y$.
The difficulty is to show stability under the transformations.

\noindent $\bullet$ Let $z=\mathsf{Simulate_w}(x)\in \I_1$, for some $w\in \{1,\dots,m\}^*$ with
$|w| \leq n_0 +1$. By ordering if we apply 
a transformation outside $\mathsf{Transfer}$ or a $\mathsf{Simulate_i}$ for some $i$, we reach
$\I_3$.
\begin{itemize}
\item For $i\leq m$, if $|w| \leq n_0$, then $\mathsf{Simulate_i}z\in \I_1$. Else,
$\mathsf{Simulate_{wi}}x=(s,c,d,n,k,p,1)$ with $|w|=n_0+1$.
But then, there exists $n_1\leqslant n_0$ such that $u^{wi}_{n_1}\neq v^{wi}_{n_1}$. Let $n_2$ be
the smallest such number, then assume without loss of generality that $c\geq d$, we have
    \begin{align*}
        s
            &= c\enc{u^{wi}}-d\enc{v^{wi}}\\
            &= (u^{wi}_{n_2}-v^{wi}_{n_2})c4^{|u^wi|-n_2}+
            \sum_{j=n_2+1}^{\max(|u^wi|,|v^wi|)}(u^{wi}_j-v^{wi}_j)c4^{|u^wi|-j}    
    \end{align*}
\text{since }$u^{wi}_j=v^{wi}_j$\text{ for }$j<n_2$. 
Thus,
\begin{align*}
        |s| &\geqslant 2c4^{|u^wi|-n_2}-\tfrac{2}{3}c4^{|u^wi|-n_2} &&\text{since }|u^{wi}_{n_2}-u^{wi}_{n_2}|=2\text{ and for } n\geq n_2, |u^{wi}_{n}-u^{wi}_{n}|\leq 2\\
            &\geqslant c4^{|u^{wi}|-n_2} \tfrac{2}{3}\\
            &\geqslant 2c + 1 &&\text{since }n_2\leqslant n_0\text{ and }|u^{wi}|\geqslant n_0+2.
    \end{align*}
    As $c\geq d$, this shows that $\mathsf{Simulate_{wi}} z\in \I_3$.  
\item $\mathsf{Transfer}z\in \I_2$. 
\end{itemize}

\noindent$\bullet$
Let $z\in \I_2$ and $f$ be one of the transformations, then $f(z) \in \I_2$ if $f$ increased
(resp. decreased) a negative (resp. positive) component. Otherwise $f(z) \in \I_3$.

\noindent $\bullet$
Let $z=(s,c,d,n,k,p,1)\in \I_3$, $f$ be one of the transformations and 
$f(z) = (s',c',d',n',k',p',1)$. 
\begin{itemize}
\item 
if $p=0$, then either $p'\leq -1$ and $f(z)\in \I_3$ or $z$ satisfies 
$(s\geq 1\vee c\geq 1\vee n\leq -1 \vee k\leq -1)$ and then 
$f(z)$ satisfies $(s'\geq 1\vee c'\geq 1\vee n'\leq -1 \vee k'\leq -1)$, thus 
$f(z)\in\I_3$.
\item if $p =1$, then $|s| -c-d \geq 1, c\geq 0 $ and $d\geq 0$. 
There is three possibilities (1) $p'=2$ and thus $f(z) \in \I_3$, (2) $f=\mathsf{Transfer}$
then $p'=0$ and either $s' \geq 1$ or $c' \geq 1$ 
and thus $f(z)\in \I_3$ or (3) $f=\mathsf{Simulate_{i}}$ for $i\leq m$. In the latter case
without loss of generality, assume that $d'\geqslant c'$ (this is completely symmetric in $c'$ and $d'$). We have that
\begin{align*}
        |s'|
            &=|\max_i(s)+c'\enc{u^{i}}-d'\enc{v^{i}}|&&\text{by applying }\mathsf{Simulate_{i}}\\
            &\geqslant \max_i|s|- d'\max(\enc{u^{i}},\enc{v^{i}})\\
            &\geqslant \max_i(c+d+1)-d'\max(\enc{u^{i}},\enc{v^{i}})&&\text{by assumption on }|s|\\
            &\geqslant \max_i(c+d+1)-\tfrac{2}{3}d\max_i
            &&\text{since }\enc{u_i}\in[0,\tfrac{2n_i}{3}]\\
            &= \max_i(c+ d/3) + \max_i  \\
            &\geqslant c'+d' + 1&&
\end{align*}
since $\max_i c\geq c'$, $\max_i d/3 \geq d'$ (as $m_i\geq 4$) and $\max_i\geqslant 4$.
This shows that $f(z)\in \I_3$.
\end{itemize}
Therefore $\I$ is inductive and thus a $\mathbb{N}$-semi-linear invariant of the system.
This concludes the reduction.\qed
\end{proof}

\section{Additional proofs for \cref{thm:1daffine}}
\subsection{Proof of \cref{lemma:twooppcounter}}

\begin{restatable}{lemma}{coprimestuff}\label{lemma:coprimestuff}
For $\ell,k$ coprime, the sequence $a_n = (n\ell \mod k)$ for $n\in\naturals$ cycles through every modulo class $\set{0,\dots,k-1}$.
\end{restatable}
\begin{proof}
Any path longer than $k$ visits some class twice, and if the shortest cycle is $k$, then it visits every class. 

Suppose there is a cycle of length less than $k$; then $n\ell = c + mk$ and $(n+i)\ell = c + m'k$ and hence $i\ell = (m'-m)k$, with $i < k$. Since $\ell$ is an integer $i$ divides $(m'-m)k$ then $i = pr$ for $p,r\in\naturals{}$ such that $\frac{m'-m}{p}$ is integer and $\frac{k}{r}$ is integer. Observe that since $r\le i< k$ we have $\frac{k}{r} > 1$. But this implies that $\frac{k}{r}$ divides $k$ and $\ell$, contradicting $\gcd(k,\ell) = 1$.
\qed \end{proof}

\twooppcounter*
\begin{proof}
Let $b = kd, c=\ell d$, where $k,\ell$ are co-prime.

We show there exists $m,n\ge 0$ such that $mb - cn = d$. We have $mb - cn = d \iff mkd -n\ell d = d \iff mk -n \ell = 1$. Then choose $m = \frac{1+n\ell}{k}$. By \cref{lemma:coprimestuff} there exists $n$ such that $n\ell$ is in any modulo class modulo $k$, and thus too for $1+n\ell$ and so $k$ divides $1+n\ell$ for some $n$.

Hence the set $\set{x +d\naturals}$ is included in the reachability set: we obtain $x+ jd$ by applying function $f$ $mj$ times and applying function $g$ $nj$ times. 
Similarly, we can find $m',n'\geq 0$ such that $m'b - cn' = -d$ and thus 
$\set{x +d\ints}$ is within the reachability set. 
\qed \end{proof}
\label{appendix:lemma:coprimestuff}

\subsection{Extended argument for non opposing counters }
\label{appendix:nonopposingcounters}

The following shows that if $\set{x + d\naturals}$ does lead to $\set{y+d\naturals}$, with $y<x$ and $y \equiv x \mod d$, then indeed we can reach $\set{z+d\naturals}$ for any $z\equiv x\mod d$ by reapplying the same set of functions which lead from $x$ to $y$.

\begin{lemma}\label{lemma:inducezcycles}Assume the reference counter $h(x) = x +d$ with $d\ge 0$. Suppose all growing functions are growing outside of $[-B,C]$. Consider $\vecit{x}{0} \in I$ and a path $\vecit{x}{0},f_{i_1},\vecit{x}{1},f_{i_2},\dots,f_{i_m},\vecit{x}{m}$ such that $\vecit{x}{j} = f_{i_j}(\vecit{x}{j-1})$, $\vecit{x}{j} \le -B$, $\vecit{x}{m} < \vecit{x}{0}$ and $\vecit{x}{0} \equiv\vecit{x}{m} \mod d$.

Then $\set{\vecit{x}{0} + d\ints} \subseteq I$.
\end{lemma}
\begin{proof}
The re-application of $f_{i_m}\circ \dots\circ f_{i_1}$ results in the same modulo class by modulo arithmetic. Further since $\vecit{x}{j} \le -B$ then any growing $f_{i_j}$, is growing by at least as much as in the first application. 
Thus $f_{i_m}\circ \dots\circ f_{i_1}(\vecit{x}{m}) < \vecit{x}{m}$.

Hence for any $M < -B$, there exists $c < M, n\in\naturals$, such that $(f_{i_m}\circ \dots\circ f_{i_1})^n(\vecit{x}{0}) = c \equiv \vecit{x}{0} \mod d$. Hence for any $\vecit{x}{0} + kd \in \set{\vecit{x}{0} + d\ints}$ with $k\in\ints$ then there exists $n,\ell$ such that $(f_{i_m}\circ \dots\circ f_{i_1})^n(\vecit{x}{0}) \le \vecit{x}{0} + kd$ and $h^\ell\circ (f_{i_m}\circ \dots\circ f_{i_1})^n(\vecit{x}{0}) = \vecit{x}{0} + kd$.\qed
\end{proof}

\begin{remark}
By symmetry, \cref{lemma:inducezcycles} also holds for the opposite direction. That is when $h(x) = x -d$, $d >0$, inequalities are inverted and $C$ is used in place of $-B$.
\end{remark}

We now consider inductively applying non-inverting functions to sets $\set{x + d\mathbb{N}} \in I$. Then add $\set{f_i(x) +d\mathbb{N}}$ provided it is not already a subset of some set already in $I$. If $\set{f_i(x) +d\mathbb{N}}$ is new and a new modulo class we can again apply \cref{lemma:inverterzs}, from whence we may also need to apply \cref{lemma:zsaturation}.

However, when this procedure does not saturate there eventually exists be a sequence of actions in which $\set{x+d\naturals}$ leads to $\set{y+d\naturals}$ with $x\equiv y\mod d$ according to a path in \cref{lemma:inducezcycles}. In particular $y < x < -B$ since if $x < y$ then $\set{y + d\naturals{}} \subseteq \set{x+d\naturals{}}$, some modulo class must repeat after at most $d$ steps, and eventually the procedure must stay $<-B$ for at least $d$ steps. Then, according to \cref{lemma:inducezcycles}, a new $\ints$-linear set can be added ($\set{x+d\ints}$) (which again can be saturated using \cref{lemma:zsaturation}). We repeat this process until all $\naturals$-linear sets are invariant. This process terminates, as each application of  \cref{lemma:inducezcycles} adds a new $\ints$-linear set with period $d$, of which there are at most $d$.

\subsection{Proof of \cref{lemma:complexityunarypoly}}
\complexityunarypoly*
\begin{proof}

In the no-counter case, by \cref{lemma:growthbound}, there is an interval $[-C,C]$ of size approximately $\frac{|b|+|M|}{|a| -1}$, where $|b|, |M|, |a|$ are all numbers represented in the input, and thus is of polynomial of size. This means the gap is of polynomial size, and thus the saturation algorithm, which must in each step add a point or terminate, is of polynomial time.

In each counter-case there is a reference counter period $d$ arising directly from the input as the counter part of some function, or in the case of two opposing counters, possibly the sum of two counter parts. For this period $d$ there are at most $3d$ possible types of non-singleton invariant ($\set{x + d\naturals{}}$ or $\set{x - d\naturals{}}$ for some $x$ and ${x + d\ints}$ for $x\in \set{0,\dots,d}$ ). Singletons only arise in the interval $[-C,C]$ if they exist. Hence, there are at most $O(2C + 3d)$ steps which change the invariant. 

In the case of two opposing counters, immediately all invariants are of the form ${x + d\ints}$ for $x\in \set{0,\dots,d}$, and the reachable modulo classes can be found in $O(dk)$ (recall $k$ is the number of functions), by breadth first search. 

In the case of all counters in the same direction, there are two phases, each has a bounded number of steps. First we consider updates which move in the direction of the counters and secondly we consider updates which move against the counters.

In the case of moving with the counters, outside of $[-C,C]$ all functions are growing. Hence, by conducting breadth first search on a priority queue that always expands the minimal element we can find the sets of the form ${x + d\naturals}$ for $x\in \set{0,\dots,d}$ in polynomial time. Only inside $[-C,C]$ does the search result in smaller elements (which there are at most $2C$ such steps), and in the remaining case we either expand to find an element already covered, or we find the smallest element in that modulo class. Thus this step takes $O(dk + 2C)$ time.

Secondly we search for cycles in the direction opposing the counters, to see if we can turn any ${x + d\naturals}$ sets into ${x + d\ints}$ sets, that is invariants induced  by \cref{lemma:inducezcycles}. There can be a path of length at most $d$ steps outside of $[-C,C]$ before a cycle is found, so the running time is $O(2Cd).$\qed
\end{proof}

\section{Tool}
\label{tool:webinterface}

The tool's output, when, applied to the MU Puzzle can be seen to produce the invariant $\set{1+3\ints} \cup \set{2+3\ints}$ of \cref{mu:puzzle}:

\begin{lstlisting}
-----------------
Interpretation of input
start: 1 target: {0} functions: [f(x) = x - 3, f(x) = 2x]
-----------------
invariant: {1 +3Z} $\cup$ {2 +3Z}
-----------------
reachability: unreachable
target {0} disjoint from invariant
-----------------
Proof of invariance
Set      under         gives        within
-------  ------------  -------  --  --------
{1 +3Z}  f(x) = x - 3  {1 +3Z}  $\ \subseteq\ $   {1 +3Z}
{1 +3Z}  f(x) = 2x     {2 +6Z}  $\ \subseteq\ $   {2 +3Z}
{2 +3Z}  f(x) = x - 3  {2 +3Z}  $\ \subseteq\ $   {2 +3Z}
{2 +3Z}  f(x) = 2x     {4 +6Z}  $\ \subseteq\ $   {1 +3Z}
-----------------
time invariant 0.000556707
time proofOfInvariant 0.000496387
-----------------
\end{lstlisting}

The web-interface can be found at \url{http://invariants.davidpurser.net}.

 \end{document}